\documentclass[aps,prl,reprint]{revtex4-2}
\usepackage{amsthm, amsmath, physics, mathtools, amssymb}
\usepackage{bm, dsfont}
\usepackage{hyperref}

\usepackage{graphicx, float}

\newtheorem{theorem}{Theorem}
\newtheorem{lemma}{Lemma}
\newtheorem{definition}{Definition}

\begin{document}
	
\title{Differential hierarchy of the Husimi representation}

\author{Zacharie Van Herstraeten}
\email{zacharie.van-herstraeten@inria.fr}
\affiliation{DIENS, École Normale Supérieure, PSL  University, CNRS, INRIA (QAT), 45 rue d’Ulm, Paris, 75005, France}

\author{Ulysse Chabaud}
\email{ulysse.chabaud@inria.fr}
\affiliation{DIENS, École Normale Supérieure, PSL  University, CNRS, INRIA (QAT), 45 rue d’Ulm, Paris, 75005, France}

\begin{abstract}
Phase-space representations reveal the geometry of quantum mechanics and distinguish quantum features from classical ones. However, the positivity of a density operator is among the properties of quantum states that phase-space representations capture least directly: deciding whether a phase-space function describes a physical quantum state has so far required global conditions coupling distant points in phase space.
Here we show that these conditions can in fact be made local, by considering the Husimi representation in phase space.
We establish a complete set of conditions which captures when a phase-space function $f$ is the Husimi function of a quantum state, with each condition involving only the derivatives of the function $f$ evaluated at a single phase-space point. We further show that each condition has an equivalent algebraic characterization, which takes the form of the nonnegativity of a determinant of the underlying operator in a displaced Fock basis. As a result, these conditions obey a recurrence relation organizing them into a differential hierarchy, which terminates based on the rank of the operator. Our results thus show that operator rank and positivity are local properties in phase space, reinforcing the role of phase-space representations as a geometric description of quantum mechanics.
\end{abstract}

\maketitle

\emph{Introduction.---}Phase space provides a common geometric language for classical mechanics and statistical physics. In Hamiltonian mechanics, the positions and momenta of a system specify a point in phase space, whose time evolution follows a Hamiltonian flow. 
In statistical mechanics, a statistical state is instead represented by a probability distribution over phase space, providing a framework in which thermodynamic behavior emerges from microscopic dynamics~\cite{Gibbs1902}.

Quantum mechanics also admits a phase-space formulation. 
The main obstacle to extending the classical picture is the noncommutativity of position and momentum, which prevents a unique correspondence between operators and functions on phase space: different operator orderings give rise to distinct representations of the same density operator~\cite{CahillGlauber1969b}. 
The Wigner function, the Glauber--Sudarshan $P$ function, and the Husimi $Q$ function are the most common phase-space quasiprobability distributions~\cite{Wigner1932,Moyal1949,Husimi1940,Glauber1963-eo,Sudarshan1963-rg}. 
These representations are not ordinary probabilities: the Wigner function can take negative values, while the $P$ function can be highly singular. Nevertheless, each is informationally complete and uniquely specifies the underlying quantum state~\cite{CahillGlauber1969b,Hillery1984}.

These phase-space representations also capture essential properties of quantum states.
Negativities and singularities provide common signatures of nonclassicality~\cite{Titulaer1965,Mandel1986,Hudson1974,Kenfack2004-cr,SperlingVogel2020}, while the optical equivalence theorem and moment hierarchies connect phase-space data to squeezing, antibunching, and entanglement~\cite{Glauber1963-eo, Sudarshan1963-rg, Shchukin2005-mt, Shchukin2005-rr,Ryl2015-za}. 
Moreover, phase-space inequalities and matrices of distributions can certify nonclassicality even when the distributions involved remain nonnegative~\cite{Bohmann2020-xq,Bohmann2020-wo}.
Beyond the distribution itself, its derivatives also carry information: local derivatives encode moments and operator matrix entries~\cite{Wunsche1991-yv,Wunsche1996-ya}, and have recently been organized into continuous-variable entanglement criteria evaluated at individual phase-space points~\cite{Callus2026-wl}, so that quantum properties are encoded not only in a distribution's profile, but in its local differential structure.

On the other hand, not all properties are easy to identify in phase space.
Perhaps the most fundamental one, the positivity of the density operator, is famously hard to capture with phase-space distributions.
Exact criteria exist, known as the KLM conditions~\cite{Kastler1965,LoupiasMiracleSole1966,LoupiasMiracleSole1967}, which characterize valid Wigner functions through constraints on their Fourier transforms~\cite{NarcowichOConnell1986}.
These constraints are global: they couple the values of the distribution at arbitrarily distant points in phase space.
Countable families of positivity tests have since been derived on phase-space lattices~\cite{Cordero2019-aj}, but they retain this nonlocal character.

In this work, we approach this problem from a different angle, by studying the validity of a phase-space distribution through the lens of its local derivatives. To do so, we focus on the Husimi representation and show that whether a phase-space distribution is the Husimi function of a genuine quantum state can be characterized solely from local differential data.

The route to our differential criterion begins with an algebraic construction. 
For an operator $\hat O$ and $n\ge0$, we consider the displaced Fock minors (Definition~\ref{def:displaced_fock_minors}):
\begin{align}
M_n(\alpha)=
\det\!\left[
\bra{i}\hat D^\dagger(\alpha)\hat O\hat D(\alpha)\ket{j}
\right]_{0\leq i,j<n},
\end{align}
where $\ket i,\ket j$ are Fock states and $\hat D(\alpha)$ is a displacement operator of amplitude $\alpha\in\mathbb C$.  
We show that the resulting hierarchy of determinants characterizes both the rank and the positive semidefiniteness of the operator as follows: the last non-vanishing minor determines the operator rank (Theorem~\ref{th:rank}), while a Hermitian operator is positive semidefinite if and only if every such minor is nonnegative at every displacement amplitude $\alpha$ (Theorem~\ref{th:psd}). 

The second step provides the differential counterpart to this algebraic characterization.
In the Husimi representation, the derivatives of $Q$ at a phase-space point encode the matrix entries of the operator displaced to that point, so that each displaced Fock minor is a finite combination of local derivatives of $Q$, a property specific to the Husimi representation.
This defines a nonlinear differential operator $\mathcal M_n$ acting on phase-space functions, with $M_n=\mathcal M_n[Q]$.
Our main result follows (Theorem~\ref{th:admissibility}): together with regularity, reality, and normalization, a function $Q$ is the Husimi function of a quantum state if and only if $\mathcal M_n[Q]\geq0$ throughout phase space for every $n$.
Positivity of the density operator, a global property, is thereby reduced to a family of local differential inequalities, identical at every point of phase space.

Finally, we uncover the mechanism that ties each of these differential operators together.
Successive Fock minors obey a recurrence relation (Theorem~\ref{th:recurrence}), which determines each of them from the two preceding ones,
\begin{align}
M_{n+1}
=
\frac{M_n^{2}}{M_{n-1}}
\left(
1+\frac{1}{2n}\Delta\ln M_n
\right),
\label{eq:intro_hierarchy}
\end{align}
where $\Delta$ is the phase-space Laplacian.
The same relation carries over to the differential operators $\mathcal M_n$.
Since the hierarchy starts with $\mathcal M_0[Q]=1$ and $\mathcal M_1[Q]=Q$, these two initial terms suffice to launch the recurrence, which then generates the entire hierarchy from $Q$ alone.

Taken together, our results do not merely show that the rank of an operator, its positivity, and ultimately whether it is a quantum state at all can be read from the local derivatives of its Husimi function: they show that these readings form a nested sequence of conditions generated recursively from $Q$, which we call the \emph{differential hierarchy of the Husimi representation}.

\medskip
\emph{Algebraic characterization.---}
We consider a single bosonic mode with Hilbert space $\mathcal{H}$, Fock basis $\{\ket n\}_{n\in\mathbb N}$, and annihilation operator $\hat a=(\hat x+i\hat p)/\sqrt{2}$, where $[\hat x,\hat p]=i$.
Phase-space points are labeled by $\alpha=(x+ip)/\sqrt{2}$, and phase-space translations are implemented by the displacement operators $\hat D(\alpha)=\exp(\alpha\hat a^\dagger-\alpha^*\hat a)$~\cite{Adesso2014-pl}.

The Fock-basis matrix entries of the displaced operator $\hat D^\dag(\alpha)\hat O\hat D(\alpha)$ depend on $\alpha$. 
To control this dependence, we work within the following class of operators.

\begin{definition}[Bargmann-entire operator]
\label{def:bargmann_entire}
An operator $\hat O$ is Bargmann-entire if its Bargmann kernel
\begin{align}
    K(z,w)
    \coloneqq
    \sum_{m,n\geq0}
    \frac{\langle m|\hat O|n\rangle}{\sqrt{m!n!}}\;
    z^m\,
    w^n
\end{align}
is jointly entire in the two independent complex variables $z$ and $w$.
\end{definition}

This is a growth condition on the Fock-matrix entries: it is satisfied by every bounded operator~\cite{Berezin1974}, hence by every density operator, and by polynomials in $\hat x$ and $\hat p$ such as the number operator.
It ensures that the displaced matrix entries $\bra i\hat D^\dagger(\alpha)\hat O\hat D(\alpha)\ket j$ are well defined and real-analytic in $\alpha$, i.e., locally given by convergent power series in $\alpha$ and $\alpha^*$ (Appendix~\ref{app:bargmann-entire-forms}, Lemma~\ref{lem:displaced-elements-kernel}).
Throughout, an operator is specified by its Fock matrix, and displaced matrix entries, rank, and positivity are understood in the generalized sense of Appendix~\ref{app:bargmann-entire-forms}, which reduces to the usual one for bounded operators.

Although unitary displacement preserves rank and positivity, the Fock-basis matrix entries of $\hat D^\dagger(\alpha)\hat O\hat D(\alpha)$ vary with $\alpha$. 
We use this variation to probe rank and positivity through a fixed $n\times n$ Fock-basis submatrix.

\begin{definition}[Displaced Fock minors]
\label{def:displaced_fock_minors}
Given an operator $\hat{O}$, we define for every $n\geq 1$ and $\alpha\in\mathbb C$
\begin{align}
    M_n[\hat O]:
    \alpha\mapsto
    \det\!\left[
        \bra{i}
        \hat D^\dagger(\alpha)\hat O\hat D(\alpha)
        \ket{j}
    \right]_{0\leq i,j<n},
    \label{eq:displaced_fock_minors}
\end{align}
together with the convention $M_0(\alpha)\coloneqq1$.
\end{definition}

For brevity, we omit the dependence on $\hat O$ and call $\{M_n\}_{n\geq0}$ the \textit{Fock-minor hierarchy}. 
The next two theorems show that this single sequence, followed throughout phase space, characterizes both rank and positivity.

\begin{theorem}[Displaced Fock minors characterize the rank]
\label{th:rank}
Let $\hat O$ be a Bargmann-entire operator with Fock minors $\lbrace M_n\rbrace$. 
The following statements are equivalent:
\begin{enumerate}
    \item[$\mathrm{(i)}$] $\operatorname{rank}(\hat O)\geq n$;
    \item[$\mathrm{(ii)}$] there exists $\alpha\in\mathbb C$ such that $M_{n}(\alpha)\neq0$;
    \item[$\mathrm{(iii)}$] $M_{n}(\alpha)\neq0$ for almost every $\alpha\in\mathbb C$.
\end{enumerate}
\end{theorem}

We sketch the proof, detailed in Appendix~\ref{app:rank-psd-theorems}.
A nonzero $n\times n$ minor immediately implies $\operatorname{rank}(\hat O)\geq n$.
Conversely, Appendix~\ref{app:wronskian_representation} expresses $M_n(\alpha)$, up to a positive factor, as a Wronskian determinant of the Bargmann kernel $K$ evaluated at $(z,w)=(\alpha^*,\alpha)$.
Just as the Wronskian of $n$ holomorphic functions vanishes identically only if they are linearly dependent~\cite{Bostan2010}, $M_n$ vanishes identically only if the Fock matrix of $\hat O$ has fewer than $n$ linearly independent columns, i.e., $\operatorname{rank}(\hat O)<n$.
Since $M_n$ is real-analytic, it otherwise vanishes only on a set of measure zero~\cite{Mityagin2020}, which yields $\mathrm{(iii)}$.
In particular, a finite-rank operator has rank $r$ precisely when $M_r\not\equiv0$ and $M_{r+1}\equiv0$.

\begin{theorem}[Displaced Fock minors characterize  positive semidefiniteness]
\label{th:psd}
A Hermitian Bargmann-entire operator $\hat O$ is positive semidefinite if and only if all its displaced Fock minors are nonnegative:
\begin{align}
    \hat O\succeq0
    \quad\Leftrightarrow\quad
    M_{n}(\alpha)\geq0
    \qquad
    \forall\,\alpha\in\mathbb C,\ 
    \forall\,n\in\mathbb{N}.
    \label{eq:psd_fock_minors}
\end{align}
\end{theorem}

We sketch the proof, given in Appendix~\ref{app:rank-psd-theorems}.
The forward implication is immediate: principal submatrices of a positive semidefinite operator are positive semidefinite.
Conversely, suppose that all minors are nonnegative, and let $r=\operatorname{rank}(\hat O)$.
By Theorem~\ref{th:rank} $\mathrm{(iii)}$, there is a displacement amplitude $\alpha_0$ at which $M_1(\alpha_0),\ldots,M_r(\alpha_0)$ are nonzero, hence strictly positive, and Sylvester's criterion makes the leading blocks of $\hat D^\dagger(\alpha_0)\,\hat{O}\,\hat D(\alpha_0)$ of order $n\leq r$ positive definite.
If $r=\infty$, this covers every block.
If $r<\infty$, a block of order $n>r$ has rank at most $r$ yet contains a positive definite block of order $r$: by Cauchy interlacing, at least $r$ of its eigenvalues are positive, which leaves no room for a negative one, and the block is positive semidefinite.
Hence $\hat D^\dagger(\alpha_0)\hat O\hat D(\alpha_0)\succeq0$, and therefore $\hat O\succeq0$.

Theorem~\ref{th:psd} should be compared with Sylvester's criterion~\cite[Sec.~7.2]{Horn2012-se}: a Hermitian $N\times N$ matrix is positive definite if and only if its $N$ leading principal minors are positive, but positive semidefiniteness requires all $2^N-1$ principal minors to be nonnegative.
Phase-space positivity criteria based on reconstructed Fock matrices involve this exponentially large family at each truncation level~\cite{Nha2008}. Instead, Theorem~\ref{th:psd} involves only the leading minors, one per order, but at every possible displacement amplitude.

For a density operator $\hat\rho$, the minors are moreover nonincreasing in $n$ and bounded by the spectrum of the state (Lemma~\ref{lem:fock_minor_bounds} in Appendix~\ref{app:rank-psd-theorems}): for all $\alpha\in\mathbb C$,
\begin{align}
    1=M_0(\alpha)
    \geq M_1(\alpha)
    \geq M_2(\alpha)
    \geq\cdots\geq0,
    \label{eq:fock_minor_monotonicity}
\end{align}
and $M_n(\alpha)\leq\prod_{k=1}^{n}\lambda_k^\downarrow(\hat\rho)$, where $\smash{\lambda_k^\downarrow(\hat\rho)}$ denotes the $k$th largest eigenvalue of $\hat\rho$.
The magnitude of each minor is thus directly constrained by the spectrum of the state.

\medskip

\emph{Differential characterization.---}
Having shown that the Fock-minor hierarchy characterizes rank and positivity, we now determine how this hierarchy is encoded in local phase-space derivatives.
The Husimi function of $\hat O$ is $Q_{\hat O}(\alpha)=\bra{\alpha}\hat O\ket{\alpha}$, where $\ket{\alpha}=\hat D(\alpha)\ket 0$; we omit the customary factor $1/\pi$, so that normalization reads $\smash{\int_{\mathbb C}\frac{\mathrm d^2\alpha}{\pi}\,Q_{\hat\rho}(\alpha)=1}$ for a state $\hat\rho$.
We use the Wirtinger derivatives~\cite{Wirtinger1927-fl} $\partial_\alpha=(\partial_x-i\partial_p)/\sqrt{2}$ and $\partial_{\alpha^*}=(\partial_x+i\partial_p)/\sqrt{2}$, which treat $\alpha$ and $\alpha^*$ as independent variables, so that the phase-space Laplacian is $\Delta=2\partial_\alpha\partial_{\alpha^*}$.

The key fact is that the derivatives of the Husimi function at the origin determine the Fock matrix of $\hat O$~\cite{Wunsche1991-yv,Wunsche1996-ya}.
Since displacing the operator translates its Husimi function, the derivatives at a point $\alpha$ likewise determine the Fock matrix of the displaced operator.
We condense this correspondence into a single operator-valued identity.

\begin{lemma}[Husimi differential reconstruction]
\label{lem:husimi-reconstruction}
Let $\hat{O}$ be a Bargmann-entire operator with Husimi function $Q_{\hat{O}}$.
For every $\alpha\in\mathbb{C}$,
\begin{align}
    \hat D^\dagger(\alpha)\,\hat{O}\,\hat D(\alpha)
    =
    e^{\partial_{\alpha^*}\hat a^\dagger}\,
    e^{\partial_\alpha\hat a}\;
    Q_{\hat O}(\alpha).
    \label{eq:husimi_reconstruction}
\end{align}
\end{lemma}

The proof is given in Appendix~\ref{app:reconstruction-formula}.
Since the exponentials are normally ordered, each Fock-basis entry of the right-hand side involves only finitely many derivatives of~$Q_{\hat O}$.
The formula extends formally to the $s$-parametrized quasiprobability distributions~\cite{CahillGlauber1969b}, but the Husimi function is the only member of the family for which finitely many derivatives suffice (Appendix~\ref{app:reconstruction-formula}).

Lemma~\ref{lem:husimi-reconstruction} turns the displaced Fock minors into differential operators acting on the Husimi function.
Taking the $n\times n$ leading determinant on its right-hand side defines the $n$th Fock-minor differential operator
\begin{align}
    \mathcal M_n[Q]:\alpha\mapsto
    \det\!\left[
    \bra{i}
    \left(
    e^{\partial_{\alpha^*}\hat a^\dagger}
    e^{\partial_\alpha\hat a}\,
    Q
    \right)
    \ket{j}
    \right]_{0\leq i,j<n},
    \label{eq:differential_fock_minor}
\end{align}
with $\mathcal M_0[Q]\equiv1$, where the differential expression acts on $Q$ before the determinant is taken.
Since only finitely many derivatives enter, $\mathcal M_n$ is defined for every smooth function $Q$ on $\mathbb C$, whether or not an underlying operator exists, and Lemma~\ref{lem:husimi-reconstruction} gives $M_n[\hat O]=\mathcal M_n[Q_{\hat O}]$ for every Bargmann-entire operator $\hat O$.
We call $\mathcal M_n[Q]$ the differential Fock minors of $Q$.
Theorem~\ref{th:psd} thereby becomes a differential criterion on phase-space functions.

\begin{theorem}[Differential characterization of valid Husimi functions]
\label{th:admissibility}
A function $Q:\mathbb C\to\mathbb C$ is the Husimi function
$Q(\alpha)=\bra{\alpha}\hat\rho\ket{\alpha}$ of a quantum state $\hat\rho$ if and only if the following conditions hold:
\begin{enumerate}
    \item[$\mathrm{(i)}$] \emph{Bargmann-entire:}
    there exists an entire function $K$ on $\mathbb C^2$ such that $Q(\alpha)=e^{-|\alpha|^2}K(\alpha^*,\alpha)$;

    \item[$\mathrm{(ii)}$] \emph{Real-valued:}
    $Q(\alpha)\in\mathbb R$ for every $\alpha\in\mathbb C$;

    \item[$\mathrm{(iii)}$] \emph{Normalized:}
    $\int_{\mathbb C}\frac{\mathrm d^2\alpha}\pi\,Q(\alpha)=1$;

    \item[$\mathrm{(iv)}$] \emph{Differentially positive:}
$\mathcal{M}_n[Q](\alpha)\geq0$ for every $\alpha\in\mathbb C$ and every $n\in\mathbb{N}$.
\end{enumerate}
\end{theorem}

Condition~$\mathrm{(i)}$ ensures that the local reconstructions obtained at different phase-space points are mutually consistent, i.e.\ they arise from a single Bargmann-entire operator whose Husimi function is $Q$.
Conditions~$\mathrm{(ii)}$ and $\mathrm{(iii)}$ encode Hermiticity and normalization, respectively.
The novel characterization is condition~$\mathrm{(iv)}$: through the reconstruction formula and Theorem~\ref{th:psd}, pointwise nonnegativity of the full Fock-minor differential hierarchy, as expressed in Eq.~\eqref{eq:differential_fock_minor}, is equivalent to positive semidefiniteness of the corresponding operator.
A proof of the full equivalence is given in Appendix~\ref{app:husimi-admissibility}.

At this stage, the differential Fock minors are defined one by one by Eq.~\eqref{eq:differential_fock_minor}, and computing them separately quickly becomes impractical as $n$ grows. 
Nonetheless, a structure already emerges at the lowest orders.
The first of them are $\mathcal M_0[Q]=1$, $\mathcal M_1[Q]=Q$, and $\mathcal M_2[Q]=Q^2+Q\,\partial_\alpha\partial_{\alpha^\ast}Q-(\partial_\alpha Q)(\partial_{\alpha^\ast}Q)$, which, wherever $Q\neq0$, can be rearranged as $\mathcal M_2[Q]=Q^2\left(1+\tfrac12\Delta\ln Q\right)$. 
Thus $\mathcal M_2$ is determined by $\mathcal M_1$ and its logarithmic curvature $\Delta\ln\mathcal M_1[Q]$, and the condition $\mathcal M_2[Q]\geq0$ reproduces the log-subharmonicity bound satisfied by physical Husimi functions~\cite{Frank2025-iu}.
This is the first instance of an exact recurrence, which relates each Fock minor to the two preceding ones.

\begin{theorem}[Differential recurrence of displaced Fock minors]
\label{th:recurrence}
Let $\hat O$ be a Bargmann-entire operator. 
For every $n\geq1$, the Fock minors satisfy
\begin{align}
    &M_n\,\partial_\alpha\partial_{\alpha^*}M_n
    -(\partial_\alpha M_n)(\partial_{\alpha^*}M_n)
    \nonumber\\
    &\qquad\qquad\qquad=
    n\left(M_{n+1}M_{n-1}-M_n^2\right).
    \label{eq:fock_minor_recurrence}
\end{align}
\end{theorem}

The proof, given in Appendix~\ref{app:differential_recurrence}, combines the recurrence relation for bi-directional Wronskians established in Ref.~\cite[Eq.~(3.1)]{Ma2011-ee} with their relation to the displaced Fock minors.
Through the identification $M_n[\hat O]=\mathcal M_n[Q_{\hat O}]$, the recurrence
carries over to the differential Fock minors, and Lemma~\ref{lem:smooth_recurrence} shows that it holds for any smooth function $Q$ under the substitution $M_k\mapsto\mathcal M_k[Q]$.
Since $\mathcal M_0[Q]=1$ and $\mathcal M_1[Q]=Q$, the recurrence generates the entire hierarchy from $Q$ alone, each step requiring derivatives of the preceding minor only up to first order in each of $\alpha$ and $\alpha^{\ast}$.
Written in bilinear form, Eq.~\eqref{eq:fock_minor_recurrence} holds throughout phase space, including where a minor vanishes; wherever two consecutive minors are nonzero, solving for the next one recovers the logarithmic form~\eqref{eq:intro_hierarchy}.

The differential recurrence from Theorem~\ref{th:recurrence} endows the Fock-minor hierarchy with an integrable dynamics.
Introducing the logarithmic coordinates
$X_n\coloneqq-\ln(M_{n+1}/M_n)$, with $X_0=-\ln Q$, the differential recurrence becomes
\begin{align}
    \partial_\alpha\partial_{\alpha^*}X_n
    =
    n\,e^{X_{n-1}-X_n}
    -(n+1)\,e^{X_n-X_{n+1}}
    +1,
    \label{eq:toda_molecule}
\end{align}
which is the equation of motion of an inhomogeneous two-dimensional Toda molecule~\cite{Leznov1983-nq,Hirota2004-tv}, with $\alpha$ and $\alpha^{\ast}$ playing the role of the two Toda times and $n$ labeling the particles of the chain. 
Positivity of every Fock minor precisely corresponds to real dynamics, in which all particle coordinates remain real. 
For an operator of rank $r$ with Husimi function $Q$, the condition $M_{r+1}\equiv0$ sends the coordinate $X_r$ of the $(r+1)$-th particle to infinity, so that the Toda chain effectively contains $r$ particles, matching the operator's rank.

For a state $\hat\rho$ of rank $r$, with spectral decomposition $\hat\rho=\sum_{j=0}^{r-1}\lambda_j\ket{\psi_j}\!\!\bra{\psi_j}$, we show in
Appendix~\ref{app:bargmann_wronskian} that the last non-vanishing Fock minor factorizes in the simple form
\begin{align}
    M_r(\alpha)
    =
    \left(\prod_{j=0}^{r-1}\frac{\lambda_j}{j!}\right)
    e^{-r|\alpha|^2}
    \left|\mathcal W_\rho(\alpha^*)\right|^2,
    \label{eq:terminal_minor_stellar_wronskian}
\end{align}
where $\mathcal W_\rho\coloneqq\operatorname{Wr}[\psi_0,\ldots,\psi_{r-1}]$ is the
Wronskian of the stellar functions of the eigenstates, i.e.\ of the entire functions
$\psi_j(z)\coloneqq \bra{0}e^{z\hat{a}}\ket{\psi_j}$~\cite{Bargmann1961-in}, which satisfy $\bra{\alpha}\ket{\psi_j}=e^{-|\alpha|^2/2}\psi_j(\alpha^{\ast})$.
This factorization separates the roles of the spectrum and the support of $\hat\rho$:
the spectrum enters only through an overall prefactor, whereas $\mathcal W_\rho$
determines the phase-space dependence.
Under a unitary change of basis within the support, $\mathcal{W}_{\rho}$ acquires only a constant phase, so that, up to the spectral prefactor, $M_r$ is an invariant of the support of $\hat{\rho}$.
Since $\mathcal W_\rho$ is entire, its zeros are isolated and form a discrete set, which
defines a stellar constellation associated with the support of $\hat\rho$.
For $r=1$, this recovers the usual stellar representation of a pure
state~\cite{Chabaud2020-mq}; at higher rank, it defines a natural stellar invariant of
the state's support.

\medskip

\emph{Conclusion.---}
In this work, we have established necessary and sufficient differential conditions for a phase-space function to be the Husimi function of a quantum state.
They combine two ingredients: displaced Fock minors, which capture the rank and positivity of an operator algebraically, and the Husimi reconstruction formula, which turns each of these minors into a finite combination of local derivatives of the Husimi function.

Along the way, we have uncovered a recurrence relation relating successive minors, which organizes the differential levels of the hierarchy and generates it entirely from the Husimi function.
This recurrence relation also provides an unexpected link to the classical integrable dynamics of the Toda molecule. We leave a more in-depth study of this connection to future work.
Interestingly, the last nonvanishing minor of a finite-rank state factorizes in terms of the stellar functions of its eigenstates, hinting at an extension of the stellar formalism to mixed states.

The Fock-minor hierarchy also suggests a way to quantify statistical mixedness in phase space, order by order. 
For a quantum state, $M_n\not\equiv0$ precisely when the state has rank at least $n$, i.e., cannot be written as a mixture of fewer than $n$ pure states, so each nonvanishing minor beyond the first signals a higher order of mixedness.
Suitable functionals of these higher-order minors could then provide higher-order extensions of the Wehrl entropy and related phase-space inequalities.

More broadly, our results delimit the set of valid Husimi functions through conditions formulated entirely within phase space, in terms of the function and its local derivatives.
Many quantum features, such as Wigner negativity, singularity of the $P$ function, or non-Gaussianity, are themselves defined in phase-space terms.
Their attainable values and tradeoffs can be studied directly in phase space, without ever reconstructing the underlying density operator.
Phase space can therefore serve not only to represent quantum states, but to determine, within its own geometry, the physical limits of the quantum features it describes.

\medskip

\emph{Acknowledgments.---}
We acknowledge funding from the European Union's Horizon Europe Framework Programme through the EIC Pathfinder Challenge project VeriQuB under Grant Agreement No.~101114899.

\medskip
\emph{AI disclosure.---}
The authors acknowledge the use of AI tools to assist with manuscript preparation and the development and verification of mathematical proofs. 
The research question and central approach were developed by the authors, who take full responsibility for the results presented in the manuscript.

\bibliography{bib.bib}

@ARTICLE{Wunsche1996-ya,
  title     = "Tomographic reconstruction of the density operator from its
               normally ordered moments",
  author    = "Wünsche, Alfred",
  journal   = "Phys. Rev. A",
  publisher = "American Physical Society",
  volume    =  54,
  number    =  6,
  pages     = "5291--5294",
  month     =  dec,
  year      =  1996,
  doi       = "10.1103/PhysRevA.54.5291",
  issn      = "1050-2947"
}

@ARTICLE{Wunsche1991-yv,
  title     = "Displaced {F}ock states and their connection to quasiprobabilities",
  author    = "Wünsche, A",
  journal   = "Quantum Opt.",
  publisher = "IOP Publishing",
  volume    =  3,
  number    =  6,
  pages     = "359--383",
  month     =  dec,
  year      =  1991,
  doi       = "10.1088/0954-8998/3/6/005",
  issn      = "0954-8998,1747-3861"
}

@ARTICLE{Bohmann2020-xq,
  title     = "Phase-space inequalities beyond negativities",
  author    = "Bohmann, Martin and Agudelo, Elizabeth",
  journal   = "Phys. Rev. Lett.",
  publisher = "American Physical Society (APS)",
  volume    =  124,
  number    =  13,
  pages     =  133601,
  month     =  apr,
  year      =  2020,
  doi       = "10.1103/PhysRevLett.124.133601",
  pmid      =  32302197,
  issn      = "0031-9007,1079-7114",
}

@ARTICLE{Frank2025-iu,
  title     = "The generalized {W}ehrl entropy bound in quantitative form",
  author    = "Frank, Rupert L and Nicola, Fabio and Tilli, Paolo",
  journal   = "J. Eur. Math. Soc.",
  publisher = "European Mathematical Society - EMS - Publishing House GmbH",
  month     =  jul,
  year      =  2025,
  doi       = "10.4171/jems/1674",
  issn      = "1435-9855,1435-9863",
  pages     = ""
}

@BOOK{Hirota2004-tv,
  title     = "Cambridge tracts in mathematics: The direct method in soliton
               theory series number 155",
  author    = "Hirota, Ryogo",
  editor    = "Nagai, Atsushi and Nimmo, Jon and Gilson, Claire",
  publisher = "Cambridge University Press",
  address   = "Cambridge, England",
  series    = "Cambridge tracts in mathematics",
  month     =  jul,
  year      =  2004,
  isbn      =  9780521836609
}

@ARTICLE{Leznov1983-nq,
  title     = "Two-dimensional exactly and completely integrable dynamical
               systems: Monopoles, instantons, dual models, relativistic
               strings, {L}und-{R}egge model, generalized {T}oda lattice, etc",
  author    = "Leznov, A N and Saveliev, M V",
  journal   = "Commun. Math. Phys.",
  publisher = "Springer Nature",
  volume    =  89,
  number    =  1,
  pages     = "59--75",
  month     =  mar,
  year      =  1983,
  doi       = "10.1007/bf01219526",
  issn      = "0010-3616,1432-0916"
}

@ARTICLE{Ma2011-ee,
  title     = "Combined {W}ronskian solutions to the {2D} {T}oda molecule equation",
  author    = "Ma, Wen-Xiu",
  journal   = "Phys. Lett. A",
  publisher = "Elsevier BV",
  volume    =  375,
  number    =  45,
  pages     = "3931--3935",
  month     =  oct,
  year      =  2011,
  doi       = "10.1016/j.physleta.2011.09.016",
  issn      = "0375-9601,1873-2429"
}

@article{Wigner1932,
  author = {Wigner, E. P.},
  title = {On the Quantum Correction For Thermodynamic Equilibrium},
  journal = {Phys. Rev.},
  volume = {40},
  pages = {749--759},
  year = {1932},
  doi = {10.1103/PhysRev.40.749}
}

@article{Moyal1949,
  author = {Moyal, J. E.},
  title = {Quantum Mechanics as a Statistical Theory},
  journal = {Math. Proc. Cambridge Philos. Soc.},
  volume = {45},
  number = {1},
  pages = {99--124},
  year = {1949},
  doi = {10.1017/S0305004100000487}
}

@article{Husimi1940,
  author = {Husimi, K.},
  title = {Some Formal Properties of the Density Matrix},
  journal = {Proc. Phys.-Math. Soc. Jpn.},
  volume = {22},
  number = {4},
  pages = {264--314},
  year = {1940},
  doi = {10.11429/ppmsj1919.22.4_264}
}

@article{CahillGlauber1969b,
  author = {Cahill, K. E. and Glauber, R. J.},
  title = {Density Operators and Quasiprobability Distributions},
  journal = {Phys. Rev.},
  volume = {177},
  pages = {1882--1902},
  year = {1969},
  doi = {10.1103/PhysRev.177.1882}
}

@article{Hillery1984,
  author = {Hillery, M. and O'Connell, R. F. and Scully, M. O. and Wigner, E. P.},
  title = {Distribution Functions in Physics: Fundamentals},
  journal = {Phys. Rep.},
  volume = {106},
  number = {3},
  pages = {121--167},
  year = {1984},
  doi = {10.1016/0370-1573(84)90160-1}
}

@article{Kastler1965,
  author = {Kastler, D.},
  title = {The {$C^*$}-Algebras of a Free Boson Field},
  journal = {Commun. Math. Phys.},
  volume = {1},
  number = {1},
  pages = {14--48},
  year = {1965},
  doi = {10.1007/BF01649588}
}

@article{LoupiasMiracleSole1966,
  author = {Loupias, G. and Miracle-Sol{\'e}, S.},
  title = {{$C^*$}-Alg{\`e}bres des Syst{\`e}mes Canoniques. {I}},
  journal = {Commun. Math. Phys.},
  volume = {2},
  pages = {31--48},
  year = {1966},
  doi = {10.1007/BF01773339}
}

@article{LoupiasMiracleSole1967,
  author = {Loupias, G. and Miracle-Sol{\'e}, S.},
  title = {{$C^*$}-Alg{\`e}bres des Syst{\`e}mes Canoniques. {II}},
  journal = {Ann. Inst. Henri Poincar{\'e} A},
  volume = {6},
  number = {1},
  pages = {39--58},
  year = {1967},
  url = {https://www.numdam.org/item/AIHPA_1967__6_1_39_0/}
}

@article{NarcowichOConnell1986,
  author = {Narcowich, F. J. and O'Connell, R. F.},
  title = {Necessary and Sufficient Conditions for a Phase-Space Function to Be a {Wigner} Distribution},
  journal = {Phys. Rev. A},
  volume = {34},
  pages = {1--6},
  year = {1986},
  doi = {10.1103/PhysRevA.34.1}
}

@article{Nha2008,
  author = {Nha, H.},
  title = {Complete Conditions for Legitimate {Wigner} Distributions},
  journal = {Phys. Rev. A},
  volume = {78},
  pages = {012103},
  year = {2008},
  doi = {10.1103/PhysRevA.78.012103}
}

@misc{Callus2026-wl,
  title         = "Revealing entanglement through local features of phase-space
                   distributions",
  author        = "Callus, Elena and Gärttner, Martin and Haas, Tobias",
  month         =  feb,
  year          =  2026,
  archivePrefix = "arXiv",
  primaryClass  = "quant-ph",
  eprint        = "2602.21688"
}

@ARTICLE{Bargmann1961-in,
  title     = "On a {H}ilbert space of analytic functions and an associated
               integral transform part {I}",
  author    = "Bargmann, V",
  journal   = "Commun. Pure Appl. Math.",
  publisher = "Wiley",
  volume    =  14,
  number    =  3,
  pages     = "187--214",
  month     =  aug,
  year      =  1961,
  doi       = "10.1002/cpa.3160140303",
  issn      = "0010-3640,1097-0312"
}

@article{Berezin1974,
  author  = {Berezin, F. A.},
  title   = {Quantization},
  journal = {Mathematics of the USSR-Izvestiya},
  volume  = {8},
  number  = {5},
  pages   = {1109--1165},
  year    = {1974},
  doi     = {10.1070/IM1974v008n05ABEH002140}
}

@book{Gibbs1902,
  author    = {Gibbs, J. Willard},
  title     = {Elementary Principles in Statistical Mechanics:
               Developed with Especial Reference to the Rational
               Foundation of Thermodynamics},
  publisher = {Charles Scribner's Sons},
  address   = {New York},
  year      = {1902}
}

@ARTICLE{Glauber1963-eo,
  title     = "Coherent and incoherent states of the radiation field",
  author    = "Glauber, Roy J",
  journal   = "Phys. Rev.",
  publisher = "American Physical Society (APS)",
  volume    =  131,
  number    =  6,
  pages     = "2766--2788",
  month     =  sep,
  year      =  1963,
  doi       = "10.1103/physrev.131.2766",
  issn      = "0031-899X,1536-6065"
}

@ARTICLE{Sudarshan1963-rg,
  title     = "Equivalence of semiclassical and quantum mechanical descriptions
               of statistical light beams",
  author    = "Sudarshan, E C G",
  journal   = "Phys. Rev. Lett.",
  publisher = "American Physical Society (APS)",
  volume    =  10,
  number    =  7,
  pages     = "277--279",
  month     =  apr,
  year      =  1963,
  doi       = "10.1103/physrevlett.10.277",
  issn      = "0031-9007,1079-7114"
}

@ARTICLE{Shchukin2005-mt,
  title     = "Nonclassicality criteria in terms of moments",
  author    = "Shchukin, E and Richter, Th and Vogel, W",
  journal   = "Phys. Rev. A",
  publisher = "American Physical Society (APS)",
  volume    =  71,
  number    =  1,
  month     =  jan,
  year      =  2005,
  doi       = "10.1103/physreva.71.011802",
  issn      = "1050-2947,1094-1622",
  pages     = "011802(R)"
}

@ARTICLE{Shchukin2005-rr,
  title     = "Nonclassical moments and their measurement",
  author    = "Shchukin, E V and Vogel, W",
  journal   = "Phys. Rev. A",
  publisher = "American Physical Society (APS)",
  volume    =  72,
  number    =  4,
  month     =  oct,
  year      =  2005,
  doi       = "10.1103/physreva.72.043808",
  issn      = "1050-2947,1094-1622",
  pages     = "043808"
}

@ARTICLE{Bohmann2020-wo,
  title     = "Probing nonclassicality with matrices of phase-space
               distributions",
  author    = "Bohmann, M and Agudelo, E and Sperling, J",
  journal   = "Quantum",
  publisher = "Verein zur Förderung des Open Access Publizierens in den
               Quantenwissenschaften",
  volume    =  4,
  pages     =  343,
  month     =  mar,
  year      =  2020,
  doi       = "10.22331/Q-2020-10-15-343",
  issn      = "2521-327X"
}

@ARTICLE{Ryl2015-za,
  title     = "Unified nonclassicality criteria",
  author    = "Ryl, S and Sperling, J and Agudelo, E and Mraz, M and Köhnke, S
               and Hage, B and Vogel, W",
  journal   = "Phys. Rev. A",
  publisher = "American Physical Society (APS)",
  volume    =  92,
  number    =  1,
  month     =  jul,
  year      =  2015,
  doi       = "10.1103/physreva.92.011801",
  issn      = "1050-2947,1094-1622",
  pages     = "011801(R)"
}

@ARTICLE{Chabaud2020-mq,
  title     = "Stellar Representation of Non-{G}aussian Quantum States",
  author    = "Chabaud, Ulysse and Markham, Damian and Grosshans, Frédéric",
  journal   = "Phys. Rev. Lett.",
  publisher = "APS",
  volume    =  124,
  number    =  6,
  pages     =  063605,
  month     =  feb,
  year      =  2020,
  doi       = "10.1103/PhysRevLett.124.063605",
  pmid      =  32109095,
  issn      = "0031-9007,1079-7114"
  }

@book{BochnerMartin1948,
  author    = {Bochner, Salomon and Martin, William Ted},
  title     = {Several {C}omplex {V}ariables},
  series    = {Princeton Mathematical Series},
  volume    = {10},
  publisher = {Princeton University Press},
  address   = {Princeton, NJ},
  year      = {1948}
}

@ARTICLE{Adesso2014-pl,
  title     = "Continuous variable quantum information: Gaussian states and
               beyond",
  author    = "Adesso, Gerardo and Ragy, Sammy and Lee, Antony R",
  journal   = "Open Syst. Inf. Dyn.",
  publisher = "World Scientific Pub Co Pte Lt",
  volume    =  21,
  number    = "01n02",
  pages     =  1440001,
  month     =  jun,
  year      =  2014,
  doi       = "10.1142/s1230161214400010",
  issn      = "1230-1612,1793-7191",
}

@ARTICLE{Foggiatto2017,
  title     = "Approximate formulas for expectation values using coherent states",
  author    = "Foggiatto, A L and Angelo, R M and Ribeiro, A D",
  journal   = "Prog. Theor. Exp. Phys.",
  publisher = "Oxford University Press (OUP)",
  volume    =  2017,
  number    =  10,
  pages     = "103A01",
  month     =  oct,
  year      =  2017,
  doi       = "10.1093/ptep/ptx129",
  issn      = "2050-3911"
}

@ARTICLE{Braunstein1998,
  title     = "Generalized phase-integrals for linear homogeneous {ODEs}",
  author    = "Braunstein, Samuel L",
  journal   = "J. Phys. A Math. Gen.",
  publisher = "IOP Publishing",
  volume    =  31,
  number    =  27,
  pages     = "5767--5773",
  month     =  jul,
  year      =  1998,
  doi       = "10.1088/0305-4470/31/27/007",
  issn      = "0305-4470,1361-6447"
}

@ARTICLE{Bostan2010,
  title     = "Wronskians and Linear Independence",
  author    = "Bostan, Alin and Dumas, Philippe",
  journal   = "Am. Math. Mon.",
  publisher = "Informa UK Limited",
  volume    =  117,
  number    =  8,
  pages     =  722,
  year      =  2010,
  doi       = "10.4169/000298910x515785",
  issn      = "0002-9890,1930-0972"
}

@ARTICLE{Mityagin2020,
  title     = "The zero set of a real analytic function",
  author    = "Mityagin, B S",
  journal   = "Math. Notes",
  publisher = "Pleiades Publishing Ltd",
  volume    =  107,
  number    = "3-4",
  pages     = "529--530",
  month     =  mar,
  year      =  2020,
  doi       = "10.1134/s0001434620030189",
  issn      = "0001-4346,1573-8876"
}

@BOOK{Hormander1990,
  title     = "An introduction to complex analysis in several variables: Volume
               7",
  author    = "Hörmander, Lars",
  publisher = "North-Holland",
  address   = "Oxford, England",
  edition   =  3,
  series    = "North-Holland Mathematical Library",
  month     =  jan,
  year      =  1990,
  isbn      =  9780444884466
}

@ARTICLE{Titulaer1965,
  title     = "Correlation functions for coherent fields",
  author    = "Titulaer, U M and Glauber, R J",
  journal   = "Phys. Rev.",
  publisher = "American Physical Society (APS)",
  volume    =  140,
  number    = "3B",
  pages     = "B676--B682",
  month     =  nov,
  year      =  1965,
  doi       = "10.1103/physrev.140.b676",
  issn      = "0031-899X,1536-6065"
}

@ARTICLE{Mandel1986,
  title     = "Non-classical states of the electromagnetic field",
  author    = "Mandel, L",
  journal   = "Phys. Scr.",
  publisher = "IOP Publishing",
  volume    = "T12",
  pages     = "34--42",
  month     =  jan,
  year      =  1986,
  doi       = "10.1088/0031-8949/1986/t12/005",
  issn      = "0031-8949,1402-4896"
}

@ARTICLE{SperlingVogel2020,
  title     = "Quasiprobability distributions for quantum-optical coherence and
               beyond",
  author    = "Sperling, J and Vogel, W",
  journal   = "Phys. Scr.",
  publisher = "IOP Publishing",
  volume    =  95,
  number    =  3,
  pages     =  034007,
  month     =  mar,
  year      =  2020,
  doi       = "10.1088/1402-4896/ab5501",
  issn      = "0031-8949,1402-4896"
}

@ARTICLE{Hudson1974,
  title   = "When is the {W}igner quasi-probability density non-negative?",
  author  = "Hudson, R L",
  journal = "Rep. Math. Phys.",
  volume  =  6,
  number  =  2,
  pages   = "249--252",
  month   =  oct,
  year    =  1974,
  doi     = "10.1016/0034-4877(74)90007-X",
  issn    = "0034-4877"
}

@ARTICLE{Kenfack2004-cr,
  title     = "Negativity of the {W}igner function as an indicator of non-classicality",
  author    = "Kenfack, Anatole and Życzkowski, Karol",
  journal   = "J. Opt. B Quantum Semiclassical Opt.",
  publisher = "IOP Publishing",
  volume    =  6,
  number    =  10,
  pages     =  396,
  month     =  aug,
  year      =  2004,
  doi       = "10.1088/1464-4266/6/10/003",
  issn      = "1464-4266",
}

@BOOK{Horn2012-se,
  title     = "Matrix Analysis",
  author    = "Horn, Roger A and Johnson, Charles R",
  publisher = "Cambridge University Press",
  month     =  oct,
  year      =  2012,
  isbn      =  9781139788885
}

@BOOK{Reed1978-so,
  title     = "{IV}: Analysis of Operators: Volume 4",
  author    = "Reed, Michael and Simon, Barry",
  publisher = "Academic Press",
  series    = "Methods of Modern Mathematical Physics",
  month     =  apr,
  year      =  1978,
  isbn      =  9780080570457
}

@ARTICLE{Wirtinger1927-fl,
  title     = "Zur formalen {T}heorie der {F}unktionen von mehr komplexen
               {V}eranderlichen",
  author    = "Wirtinger, W",
  journal   = "Math. Ann.",
  publisher = "Springer Science and Business Media LLC",
  volume    =  97,
  number    =  1,
  pages     = "357--375",
  month     =  dec,
  year      =  1927,
  doi       = "10.1007/bf01447872",
  issn      = "0025-5831,1432-1807"
}

@ARTICLE{Cordero2019-aj,
  title     = "On the positivity of trace class operators",
  author    = "Cordero, Elena and de Gosson, Maurice and Nicola, Fabio",
  journal   = "Adv. Theor. Math. Phys.",
  publisher = "International Press of Boston",
  volume    =  23,
  number    =  8,
  pages     = "2061--2091",
  year      =  2019,
  doi       = "10.4310/atmp.2019.v23.n8.a4",
  issn      = "1095-0761,1095-0753"
}

\newpage

\appendix
\setcounter{secnumdepth}{1}

\begin{center}
    {\Large Appendix}
\end{center}

\section{Bargmann-entire operators}
\label{app:bargmann-entire-forms}

Throughout this work, the term operator refers to a Fock-basis matrix, which generalizes the notion of a Hilbert-space operator.
Let $\mathcal H_{\mathrm{fin}}\coloneqq\operatorname{span}\{\ket n:n\geq0\}$ be the space of finite Fock superpositions.
A Fock-basis matrix is a complex matrix $O=(O_{mn})_{m,n\geq0}$, viewed as a sesquilinear form on $\mathcal H_{\mathrm{fin}}$ through
\begin{align}
O[\phi,\psi]\coloneqq\sum_{m,n}\phi_m^\ast O_{mn}\psi_n.
\end{align}
Since both arguments have finite Fock support, this sum is always finite.

Every Hilbert-space operator whose domain contains $\mathcal H_{\mathrm{fin}}$ induces such a Fock-basis matrix, but a Fock-basis matrix does not necessarily extend to an operator on $\mathcal H$ (e.g., $O_{mn}=1$ for all $m,n$).
Nonetheless, throughout this work, we refer to Fock-basis matrices as operators and denote them by $\hat O$.
Whenever the distinction is needed, we explicitly use the term Hilbert-space operator.

Common properties of Hilbert-space operators extend directly to Fock-basis matrices.
The adjoint has matrix entries $(O^\dagger)_{mn}=O_{nm}^*$.
The Fock-basis matrix is Hermitian when $O^\dagger=O$ and positive semidefinite when $O[\psi,\psi]\geq0$ for every $\psi\in\mathcal H_{\mathrm{fin}}$.
Its rank is $\operatorname{rank}\hat O\coloneqq\sup_{N\geq1}\operatorname{rank}[O_{mn}]_{m,n=0}^{N-1}$.
For a positive Fock-basis matrix, we also define the (possibly infinite) trace by $\operatorname{Tr}(\hat O)\coloneqq\sum_{n\geq0}O_{nn}$.
For bounded Hilbert-space operators, these definitions agree with their usual operator-theoretic counterparts.

Extending a Fock-basis matrix beyond $\mathcal H_{\mathrm{fin}}$ involves handling infinite sums that need not converge. 
This is relevant in this work because we consider displaced Fock states, which generally lie outside $\mathcal H_{\mathrm{fin}}$. 
To control these sums, we use the associated Bargmann kernel \begin{align}
K(z,w)
\coloneqq
\sum_{p,q\geq0}
\frac{O_{pq}}{\sqrt{p!q!}}\;z^p\,w^q,
\end{align}
where $z$ and $w$ are independent complex variables. As in Definition~\ref{def:bargmann_entire}, a Fock-basis matrix is Bargmann-entire when its Bargmann kernel is jointly entire on $\mathbb C^2$. 
To relate this kernel to displaced Fock matrix entries, we introduce the unnormalized coherent states:
\begin{align}
\lVert w\rangle
&\coloneqq e^{w\hat a^\dagger}\ket0
=\sum_{n\geq0}\frac{w^n}{\sqrt{n!}}\ket n,
\\
\langle z\rVert
&\coloneqq\bra0e^{z\hat a}
=\sum_{m\geq0}\frac{z^m}{\sqrt{m!}}\bra m.
\end{align}
For a Bargmann-entire operator $\hat O$, the kernel can then be written compactly as $K(z,w)=\langle z\rVert\hat O\lVert w\rangle$.
Hereafter, we use the displacement identities
\begin{align}
\hat D(\alpha)\lVert w\rangle
&=e^{-|\alpha|^2/2-\alpha^*w}\lVert w+\alpha\rangle,
\label{eq:displacement-unnorm-coherent-1}
\\[0.6em]
\langle z\rVert\hat D^\dagger(\alpha)
&=e^{-|\alpha|^2/2-\alpha z}\langle z+\alpha^*\rVert
\label{eq:displacement-unnorm-coherent-2}
\end{align}
to establish absolute convergence of the displaced Fock matrix entries and to express them in terms of $K$.

\begin{lemma}[Absolute convergence of displaced Fock matrix entries]
\label{lem:displaced-elements}
Let $\hat O$ be a Bargmann-entire operator, with Fock-basis matrix $(O_{pq})_{p,q\geq0}$. Then, for every $\alpha\in\mathbb C$ and $m,n\geq0$, the series
\begin{align}
O_{mn}^{(\alpha)}
\coloneqq
\sum_{p,q\geq0}
\bra m\hat D^\dagger(\alpha)\ket p\,
O_{pq}\,
\bra q\hat D(\alpha)\ket n
\label{eq:def-displaced-fock-matrix-element}
\end{align}
converges absolutely, and thus defines a valid Fock-basis matrix $O^{(\alpha)}$.
\end{lemma}

\begin{proof}
We first bound the displacement inner products using Cauchy's estimate, then use the absolute convergence of the Taylor series of $K$.

For each $q\geq0$, define the entire function
\begin{align}
d_q(w)
&\coloneqq\bra q\hat D(\alpha)\lVert w\rangle
\nonumber\\
&=\sum_{k\geq0}\frac{w^k}{\sqrt{k!}}
\bra q\hat D(\alpha)\ket k
\nonumber\\
&=e^{-|\alpha|^2/2-\alpha^*w}
\frac{(w+\alpha)^q}{\sqrt{q!}}.
\end{align}
Its $n$th derivative at zero is $d_q^{(n)}(0)=\sqrt{n!}\bra q\hat D(\alpha)\ket n$.
Cauchy's estimate on the unit circle, $|d_q^{(n)}(0)|\leq n!\sup_{|w|=1}|d_q(w)|$, together with the expression of $d_q$, gives
\begin{align}
\left|\bra q\hat D(\alpha)\ket n\right|
&\leq\sqrt{n!}\sup_{|w|=1}|d_q(w)|
\nonumber\\
&\leq\sqrt{n!}\,e^{-|\alpha|^2/2+|\alpha|}
\frac{(1+|\alpha|)^q}{\sqrt{q!}}
\nonumber\\
&=C_n\frac{R^q}{\sqrt{q!}},
\end{align}
where $R=1+|\alpha|$ and $C_n=\sqrt{n!}\,e^{-|\alpha|^2/2+|\alpha|}$ is independent of $q$.

Applying the same bound to $\bra p\hat D(\alpha)\ket m$, we obtain
\begin{align}
\nonumber&\sum_{p,q\geq0}
\left|
\bra m\hat D^\dagger(\alpha)\ket p
O_{pq}
\bra q\hat D(\alpha)\ket n
\right|
\\
&\qquad\qquad\leq C_mC_n
\sum_{p,q\geq0}
\frac{|O_{pq}|R^{p+q}}{\sqrt{p!q!}}
<\infty.
\end{align}
This last sum is finite because the Taylor series of the jointly entire kernel $K$ converges absolutely at $(R,R)$~\cite[Theorem~2.2.6]{Hormander1990}.
\end{proof}

We denote the Fock-basis matrix $O^{(\alpha)}$ with entries $O_{mn}^{(\alpha)}$ by $\hat D^\dagger(\alpha)\hat O\hat D(\alpha)$.
For bounded Hilbert-space operators, these are the usual displaced matrix entries.
We next express these matrix entries directly in terms of $K$.

\begin{lemma}[Displaced Fock matrix entries from the Bargmann kernel]
\label{lem:displaced-elements-kernel}
Let $\hat O$ be a Bargmann-entire operator with Bargmann kernel $K$.
For every $\alpha\in\mathbb C$ and $m,n\geq0$,
\begin{align}
    O_{mn}^{(\alpha)}
    ={}&\frac{e^{|\alpha|^2}}{\sqrt{m!n!}}
    \left.
    \partial_z^m\partial_w^n
    \left[e^{-\alpha z-\alpha^*w}K(z,w)\right]
    \right|_{z=\alpha^*,\,w=\alpha}.
    \label{eq:displaced-fock-from-kernel}
\end{align}
In particular, $\alpha\mapsto O_{mn}^{(\alpha)}$ is real-analytic in $(\operatorname{Re}\alpha,\operatorname{Im}\alpha)$.
\end{lemma}

\begin{proof}
Applying the displacement identities \eqref{eq:displacement-unnorm-coherent-1}--\eqref{eq:displacement-unnorm-coherent-2} to $K(z,w)=\langle z\Vert\hat{O}\Vert w\rangle$ gives
\begin{align}
&\langle z\rVert
\hat D^\dagger(\alpha)\hat O\hat D(\alpha)
\lVert w\rangle
\nonumber
\\
&\qquad=
e^{-|\alpha|^2/2-\alpha z}
e^{-|\alpha|^2/2-\alpha^\ast w}
\langle z+\alpha^{\ast}\rVert
\hat O
\lVert w+\alpha\rangle
\nonumber
\\
&\qquad=
e^{-|\alpha|^2-\alpha z-\alpha^\ast w}
K(z+\alpha^\ast,w+\alpha).
\label{eq:unnormalized-coherent-to-kernel}
\end{align}

On the other hand, expanding in the Fock basis and using the definition \eqref{eq:def-displaced-fock-matrix-element} gives
\begin{align}
&\langle z\rVert
\hat D^\dagger(\alpha)\hat O\hat D(\alpha)
\lVert w\rangle
\nonumber
\\
&\qquad=
\sum\limits_{p,q\geq 0}
\langle z\rVert
\hat D^\dagger(\alpha)\ket{p}
O_{pq}
\bra{q}\hat D(\alpha)
\lVert w\rangle
\nonumber
\\&\qquad=
\sum\limits_{m,n\geq 0}
\frac{z^m w^n}{\sqrt{m!n!}}
\sum\limits_{p,q\geq 0}
\langle m|
\hat D^\dagger(\alpha)\ket{p}
O_{pq}
\bra{q}\hat D(\alpha)
| n\rangle
\nonumber
\\&\qquad=
\sum_{m,n\geq0}
\frac{O_{mn}^{(\alpha)}}{\sqrt{m!n!}}\,z^mw^n.
\label{eq:unnormalized-coherent-to-sum-matrix-element}
\end{align}

Combining Eqs.~\eqref{eq:unnormalized-coherent-to-kernel} and~\eqref{eq:unnormalized-coherent-to-sum-matrix-element} allows us to recover $O_{mn}^{(\alpha)}$ from $K$.
Differentiating $m$ times with respect to $z$ and $n$ times with respect to $w$, evaluating at the origin, and dividing by $\sqrt{m!n!}$ gives
\begin{align}
O_{mn}^{(\alpha)}
\!=\!
\frac{e^{-|\alpha|^2}}{\sqrt{m!n!}}
\left.
\partial_z^m\partial_w^n
\left[
e^{-\alpha z-\alpha^\ast w}
K(z\!+\!\alpha^\ast,w\!+\!\alpha)
\right]
\right|_{z=w=0}\!.
\label{eq:displaced-fock-coefficient}
\end{align}

Under the translation $z'=z+\alpha^\ast$ and $w'=w+\alpha$, the evaluation point becomes $(\alpha^\ast,\alpha)$ and the exponential factor in the bracket becomes $e^{2|\alpha|^2-\alpha z'-\alpha^\ast w'}$, which yields Eq.~\eqref{eq:displaced-fock-from-kernel}.

Finally, to establish real-analyticity, we regard $\alpha^\ast$ and $\alpha$ in Eq.~\eqref{eq:displaced-fock-coefficient} as independent complex variables $u,v$.
The expression $e^{-uv-vz-uw}K(z+u,w+v)$ is jointly entire in $(u,v,z,w)$.
Differentiation with respect to $z,w$ and evaluation at $z=w=0$ therefore yield an entire function of $(u,v)$.
Restricting to $(u,v)=(\alpha^\ast,\alpha)$ therefore implies that $\alpha\mapsto O_{mn}^{(\alpha)}$ is real-analytic in $(\operatorname{Re}\alpha,\operatorname{Im}\alpha)$.
\end{proof}

In particular, setting $m=n=0$ in Eq.~\eqref{eq:displaced-fock-from-kernel} gives
$Q(\alpha)=O_{00}^{(\alpha)}=e^{-|\alpha|^2}K(\alpha^\ast,\alpha)$.
Moreover, every finite minor of the displaced Fock matrix is real-analytic, since it is the determinant of a matrix with real-analytic entries.

\section{Phase-space reconstruction formula}
\label{app:reconstruction-formula}

In this appendix, we prove the Husimi reconstruction formula below for Bargmann-entire operators, and discuss a formal extension to other quasiprobability distributions in the Cahill--Glauber family~\cite{CahillGlauber1969b}.
All identities are understood entrywise in the Fock basis, with $\alpha$ and $\alpha^\ast$ treated as independent variables for differentiation.

\begin{proof}[Proof of Lemma~\ref{lem:husimi-reconstruction}]
Normal ordering makes the right-hand side well defined: in its $(m,n)$ matrix entry, only powers of $\hat a^\dagger$ up to $m$, and of $\hat a$ up to $n$ contribute.

By Eq.~\eqref{eq:unnormalized-coherent-to-kernel}, the Bargmann kernel of the displaced operator satisfies
\begin{align}
\langle z\rVert
\hat D^\dagger(\alpha)\hat O\hat D(\alpha)
\lVert w\rangle
&=
e^{-|\alpha|^2-\alpha z-\alpha^\ast w}
K(z+\alpha^\ast,w+\alpha)
\nonumber\\
&=
e^{zw}e^{z\partial_{\alpha^\ast}}
e^{w\partial_\alpha}Q(\alpha).
\label{eq:husimi-generating-identity}
\end{align}
The second equality uses $Q(\alpha)=e^{-|\alpha|^2}K(\alpha^\ast,\alpha)$.
Since $e^{-uv}K(u,v)$ is entire, these translations converge for all $z,w$.

Finally, $\langle z\rVert\hat a^\dagger=z\langle z\rVert$, $\hat a\lVert w\rangle=w\lVert w\rangle$, and $\langle z\rVert w\rangle=e^{zw}$, so the last expression in Eq.~\eqref{eq:husimi-generating-identity} is also the Bargmann kernel of the right-hand side of Eq.~\eqref{eq:husimi_reconstruction}.
The equality of the Taylor coefficients proves the result.
\end{proof}

The Husimi function is the $s=-1$ member of the Cahill--Glauber family $W^{(s)}$, which also includes the Wigner ($s=0$) and Glauber--Sudarshan ($s=1$) representations.
We now consider a formal extension of its reconstruction formula to other members of this family.
For $s>-1$, Gaussian smoothing gives~\cite{CahillGlauber1969b}
\begin{align}
Q(\alpha)
=
\frac{2}{s+1}
\int_{\mathbb C}\frac{\mathrm d^2\beta}{\pi}\,
e^{-\frac{2}{s+1}|\alpha-\beta|^2}W^{(s)}(\beta).
\label{eq:cahill-glauber-convolution}
\end{align}
Thus reconstruction from $W^{(s)}$ can be achieved formally by recovering $Q$ via a Gaussian smoothing and then applying the Husimi reconstruction formula in Lemma~\ref{lem:husimi-reconstruction}.

The same smoothing can formally be expressed as an exponential of the Laplacian (see \cite[App.~A, Eqs.~(A.4)--(A.6)]{Foggiatto2017} for a derivation and a discussion of convergence):
\begin{align}
Q
=e^{\frac{s+1}{2}\partial_\alpha\partial_{\alpha^\ast}}
W^{(s)}.
\label{eq:gaussian-differential-operator}
\end{align}

Since $[\hat a\partial_\alpha,\hat a^\dagger\partial_{\alpha^\ast}]=\partial_\alpha\partial_{\alpha^\ast}$ commutes with both $\hat a\partial_\alpha$ and $\hat a^\dagger\partial_{\alpha^\ast}$, the Baker--Campbell--Hausdorff formula allows us to combine formally Eqs.~\eqref{eq:husimi_reconstruction} and \eqref{eq:gaussian-differential-operator} as
\begin{align}
\hat D^\dagger(\alpha)\hat O\hat D(\alpha)
&=
e^{\partial_{\alpha^\ast}\hat a^\dagger}
e^{\partial_\alpha\hat a}Q(\alpha)
\nonumber\\
&=
e^{\partial_{\alpha^\ast}\hat a^\dagger}
e^{\partial_\alpha\hat a}
e^{\frac{s+1}{2}\partial_\alpha\partial_{\alpha^\ast}}
W^{(s)}(\alpha)
\nonumber\\
&=
\exp\!\left(
\partial_\alpha\hat a
+\partial_{\alpha^\ast}\hat a^\dagger
+\tfrac{s}{2}\partial_\alpha\partial_{\alpha^\ast}
\right)W^{(s)}(\alpha).
\label{eq:formal-s-reconstruction}
\end{align}
The first equality is the Husimi reconstruction formula proven above.
Note that the Gaussian differential expansion and the final single exponential are formal expressions whose power-series convergence is not guaranteed in general. 

Finally, let us highlight that the Husimi function is in fact the only member of the family for which the reconstruction involves finitely many derivatives.
For $s\neq-1$, the factor $e^{(s+1)\partial_\alpha\partial_{\alpha^\ast}/2}$ in Eq.~\eqref{eq:formal-s-reconstruction} is a differential operator of infinite order, and it cannot be traded for one of finite order.
To see this, consider a coherent state $\ket\beta$: for $s<1$, its distribution $W^{(s)}$ is proportional to $e^{-2|\alpha-\beta|^2/(1-s)}$, a Gaussian whose width depends on $s$, while $Q(\alpha)=e^{-|\alpha-\beta|^2}$; for $s\geq1$, $W^{(s)}$ is not even a function.
A differential operator of finite order multiplies a Gaussian by a polynomial but cannot change its width, so no such operator maps $W^{(s)}$ to $Q$ unless $s=-1$.

\section{Wronskian representation of displaced Fock minors}
\label{app:wronskian_representation}

In this appendix, we relate displaced Fock minors to Wronskian determinants of the Bargmann kernel.
We then show that the determinant of order $n$ is not identically zero if and only if the Fock matrix has rank at least $n$.
These results will be used to prove the rank characterization in Appendix~\ref{app:rank-psd-theorems} and the differential recurrence in Appendix~\ref{app:differential_recurrence}.

For $n\geq1$ and holomorphic functions $f_0,\ldots,f_{n-1}$ of one complex variable $z$, we denote their Wronskian by
\begin{align}
    \operatorname{Wr}[f_0,\ldots,f_{n-1}](z)
    \coloneqq
    \det\!\left[
        \partial_z^i f_j(z)
    \right]_{i,j=0}^{n-1}.
\end{align}
For functions of two variables, we use the bi-directional Wronskian determinant~\cite{Ma2011-ee}, defined as follows.

\begin{definition}[Bi-directional Wronskian]
\label{def:biwronskian}
For a jointly entire function $F(z,w)$ of two independent complex variables, define
\begin{align}
    \tau_n[F](z,w)
    \coloneqq
    \det\!\left[
        \partial_z^i\partial_w^jF(z,w)
    \right]_{i,j=0}^{n-1},
    \qquad n\geq1,
\end{align}
with $\tau_0[F]\coloneqq1$.
\end{definition}

The following lemma gives the precise relation between the bi-directional Wronskians of the Bargmann kernel and the displaced Fock minors.

\begin{lemma}[Wronskian representation of displaced Fock minors]
\label{lem:fock_minors_biwronskians}
Let $\hat O$ be a Bargmann-entire operator with Bargmann kernel $K$.
Then, for every $n\geq1$ and every $\alpha\in\mathbb C$,
\begin{align}
    \tau_n[K](\alpha^*,\alpha)
    =
    \left(
        \prod_{k=0}^{n-1}k!
    \right)
    e^{n|\alpha|^2}
    M_{n}(\alpha).
    \label{eq:M_tau_relation}
\end{align}
\end{lemma}

\begin{proof}
For fixed $\alpha$, introduce the modified kernel
\begin{align}
    \widetilde K_\alpha(z,w)
    \coloneqq
    e^{-\alpha z-\alpha^*w}K(z,w).
\end{align}
Taking the determinant in Eq.~\eqref{eq:displaced-fock-from-kernel} gives
\begin{align}
    M_{n}(\alpha)
    =
    \frac{e^{n|\alpha|^2}}
         {\prod_{k=0}^{n-1}k!}\,
    \tau_n[\widetilde K_\alpha](\alpha^*,\alpha).
    \label{eq:minor_modified_biwronskian}
\end{align}
Applying the multiplicative Wronskian identity~\cite[Eq.~(14)]{Braunstein1998} successively in $z$ and $w$ yields
\begin{align}
    \tau_n[\widetilde K_\alpha](z,w)
    =
    e^{-n(\alpha z+\alpha^*w)}
    \tau_n[K](z,w).
\end{align}
Evaluating at $(z,w)=(\alpha^*,\alpha)$ and substituting into Eq.~\eqref{eq:minor_modified_biwronskian} proves Eq.~\eqref{eq:M_tau_relation}.
\end{proof}

We now show that the nonvanishing of these determinants characterizes the rank of the Fock matrix.

\begin{lemma}[Rank characterization by bi-directional Wronskians]
\label{lem:bargmann_biwronskian_rank}
Let $\hat O$ be a Bargmann-entire operator with Bargmann kernel $K$.
For every $n\geq1$,
\begin{align}
    \operatorname{rank}(\hat O)\geq n
    \quad\Leftrightarrow\quad
    \tau_n[K]\not\equiv0.
\end{align}
\end{lemma}

\begin{proof}
Write $O_{pq}=\langle p|\hat O|q\rangle$ and, for every $q\geq0$, define the entire function
\begin{align}
    f_q(z)
    \coloneqq
    \frac{1}{\sqrt{q!}}\,\partial_w^qK(z,0)
    =
    \sum_{p\geq0}
    \frac{O_{pq}}{\sqrt{p!}}\,z^p.
\end{align}
Linear relations over $\mathbb C$ among the functions $f_q$ are exactly linear relations among the corresponding columns of the Fock matrix.
Consequently, $\operatorname{rank}(\hat O)\geq n$ if and only if some $n$ functions $f_{q_1},\ldots,f_{q_n}$ are linearly independent.

Suppose first that $\operatorname{rank}(\hat O)\geq n$, and choose such a family.
By the Wronskian criterion for holomorphic functions~\cite{Bostan2010}, there exists $z_0\in\mathbb C$ such that
\begin{align}
    \det\!\left[
        \partial_z^i f_{q_j}(z_0)
    \right]_{\substack{0\leq i<n\\1\leq j\leq n}}
    \neq0.
    \label{eq:independent-column-wronskian}
\end{align}
For this $z_0$, define the entire functions
\begin{align}
    h_i(w)\coloneqq\partial_z^iK(z_0,w),
    \qquad 0\leq i<n.
\end{align}
Since
\begin{align}
    \frac{1}{\sqrt{q!}}\,\partial_w^q h_i(0)
    =
    \partial_z^i f_q(z_0),
    \label{eq:column-function-derivatives}
\end{align}
the nonzero determinant in Eq.~\eqref{eq:independent-column-wronskian} implies that $h_0,\ldots,h_{n-1}$ are linearly independent.
Applying the same Wronskian criterion to this family gives a point $w_0\in\mathbb C$ such that
\begin{align}
    0
    &\neq
    \det\!\left[
        \partial_w^j h_i(w_0)
    \right]_{i,j=0}^{n-1}
    \nonumber\\
    &=
    \det\!\left[
        \partial_z^i\partial_w^jK(z_0,w_0)
    \right]_{i,j=0}^{n-1}
    =
    \tau_n[K](z_0,w_0).
\end{align}
Thus $\tau_n[K]\not\equiv0$.

Conversely, suppose that $\tau_n[K](z_0,w_0)\neq0$ for some $(z_0,w_0)\in\mathbb C^2$.
Then the functions $h_i(w)=\partial_z^iK(z_0,w)$, $0\leq i<n$, are linearly independent.
Since these functions are entire, the matrix of their Taylor coefficients at $w=0$ has rank $n$.
Using Eq.~\eqref{eq:column-function-derivatives}, we can therefore choose indices $q_1,\ldots,q_n$ for which Eq.~\eqref{eq:independent-column-wronskian} holds.
This implies that $f_{q_1},\ldots,f_{q_n}$ are linearly independent, as are the corresponding columns of the Fock matrix.
Hence $\operatorname{rank}(\hat O)\geq n$.
\end{proof}

\section{Rank and positivity from displaced Fock minors}
\label{app:rank-psd-theorems}

In this appendix, we prove Theorems~\ref{th:rank} and~\ref{th:psd}, which characterize the rank and positivity of Bargmann-entire operators through their displaced Fock minors.
We begin with the rank characterization, using the Wronskian results from Appendix~\ref{app:wronskian_representation} and the real-analyticity of the minors established in Appendix~\ref{app:bargmann-entire-forms}.

\begin{proof}[Proof of Theorem~\ref{th:rank}]
Fix $n\geq1$, and let $K$ be the Bargmann kernel of $\hat O$.
By Lemma~\ref{lem:bargmann_biwronskian_rank}, condition $\mathrm{(i)}$ is equivalent to $\tau_n[K]\not\equiv0$.
An entire function of two complex variables that vanishes at every point $(\alpha^*,\alpha)$, $\alpha\in\mathbb C$, must vanish identically~\cite[Proposition II.4.7]{BochnerMartin1948}.
Together with Lemma~\ref{lem:fock_minors_biwronskians}, this gives
\begin{align}
    \operatorname{rank}(\hat O)\geq n
    &\quad\Leftrightarrow\quad
    \tau_n[K]\not\equiv0
    \nonumber\\
    &\quad\Leftrightarrow\quad
    M_{n}\not\equiv0.
\end{align}
Thus $\mathrm{(i)}$ and $\mathrm{(ii)}$ are equivalent.

We next prove $\mathrm{(ii)}\Rightarrow\mathrm{(iii)}$.
By Appendix~\ref{app:bargmann-entire-forms}, the function $\alpha\mapsto M_{n}(\alpha)$ is real-analytic.
Under condition $\mathrm{(ii)}$, its squared modulus is a real-analytic function that is not identically zero.
Its zero set therefore has Lebesgue measure zero~\cite{Mityagin2020}.
Since $M_{n}$ and its squared modulus have the same zeros, condition $\mathrm{(iii)}$ follows.

Finally, $\mathrm{(iii)}$ implies $\mathrm{(ii)}$, completing the proof.
\end{proof}

We now use the rank characterization to establish the positivity criterion.

\begin{proof}[Proof of Theorem~\ref{th:psd}]
From Appendix~\ref{app:bargmann-entire-forms}, recall the notation $O^{(\alpha)}_{mn}=\bra{m}\hat{D}^{\dagger}(\alpha)\hat{O}\hat{D}(\alpha)\ket{n}$.
For $N\geq1$ and $\alpha\in\mathbb C$, define the $N\times N$ matrix:
\begin{align}
    B_N(\alpha)
    \coloneqq
    \left[
        O_{ij}^{(\alpha)}
    \right]_{i,j=0}^{N-1},
\end{align}
so that $M_N(\alpha)=\det B_N(\alpha)$.
For $L\geq1$, also set
\begin{align}
    O_L
    &\coloneqq
    \left[O_{pq}\right]_{p,q=0}^{L-1},
    \\
    T_{L,N}(\alpha)
    &\coloneqq
    \left[
        \bra q\hat D(\alpha)\ket j
    \right]_{\substack{0\leq q<L\\0\leq j<N}}.
\end{align}
By Lemma~\ref{lem:displaced-elements},
\begin{align}
    B_N(\alpha)
    =
    \lim_{L\to\infty}
    T_{L,N}(\alpha)^\dagger\;
    O_L\;
    T_{L,N}(\alpha),
    \label{eq:displaced-block-truncation}
\end{align}
where the limit holds entrywise.
Since $\hat O$ is Hermitian, all matrices $B_N(\alpha)$ are Hermitian.

Suppose first that $\hat O\succeq0$.
Then every $O_L$ is positive semidefinite, as is each matrix on the right-hand side of Eq.~\eqref{eq:displaced-block-truncation}.
Passing to the limit gives $B_N(\alpha)\succeq0$, and hence $M_N(\alpha)\geq0$ for every $N$ and $\alpha$.

Conversely, suppose that $M_N(\alpha)\geq0$ for every $N$ and $\alpha$, and let $r=\operatorname{rank}(\hat O)$.
If $r=0$, then $\hat O=0$ and the conclusion is immediate.
We therefore assume $r>0$.

By Theorem~\ref{th:rank}, the zero set of $M_n$ has Lebesgue measure zero for every $1\leq n\leq r$.
Since there are at most countably many such zero sets, their union has measure zero, and its complement $E\subset\mathbb C$ has full measure.
On $E$, all these minors are nonzero and hence strictly positive, since they are nonnegative by assumption:
\begin{align}
    M_n(\alpha)>0
    \qquad
    \forall\,\alpha\in E,
    \quad
    \forall\,1\leq n\leq r.
    \label{eq:positive_principal_minors}
\end{align}

Suppose first that $r<\infty$.
Each matrix on the right-hand side of Eq.~\eqref{eq:displaced-block-truncation} has rank at most $r$.
This bound is preserved in the limit, since all minors of order $r+1$ vanish and determinants depend continuously on the matrix entries.
Hence $\operatorname{rank}B_N(\alpha)\leq r$ for every $N$ and $\alpha$.

Fix $\alpha\in E$ and write $B_N=B_N(\alpha)$.
By Sylvester's criterion, $B_r\succ0$.
For $N>r$, the Hermitian matrix $B_N$ has rank at most $r$ and contains $B_r\succ0$ as a principal submatrix.
By Cauchy's interlacing theorem~\cite[Thm.~4.3.17]{Horn2012-se}, its $r$ largest eigenvalues are bounded below by those of $B_r$, and are therefore strictly positive.
Since $B_N$ has at most $r$ nonzero eigenvalues, all its remaining eigenvalues vanish, and $B_N\succeq0$.

Suppose now that $r=\infty$, and fix $N\geq1$ and $\alpha\in E$.
The leading principal minors of $B_N(\alpha)$ are $M_1(\alpha),\ldots,M_N(\alpha)$, all of which are strictly positive by Eq.~\eqref{eq:positive_principal_minors}.
Sylvester's criterion therefore gives $B_N(\alpha)\succ0$.

In both cases, $B_N(\alpha)\succeq0$ for every $N\geq1$ and every $\alpha\in E$.
Since $E$ is dense and each $B_N$ is continuous by Appendix~\ref{app:bargmann-entire-forms}, it follows that
\begin{align}
    B_N(0)\succeq0
    \qquad
    \forall\,N\geq1.
\end{align}
Every finite Fock superposition is supported in one of these initial blocks.
Hence $O[\psi,\psi]\geq0$ for every $\psi\in\mathcal H_{\mathrm{fin}}$, proving $\hat O\succeq0$.
\end{proof}

We finally prove the spectral bounds quoted in the main text.

\begin{lemma}[Monotonicity and spectral bounds]
\label{lem:fock_minor_bounds}
Let $\hat\rho$ be a density operator with displaced Fock minors $\{M_n\}_{n\geq0}$, and let $\lambda_1^\downarrow(\hat\rho)\geq\lambda_2^\downarrow(\hat\rho)\geq\cdots$ denote its eigenvalues in nonincreasing order, counted with multiplicity and extended by zeros if $\hat\rho$ has finite rank.
For every $\alpha\in\mathbb C$,
\begin{align}
    1=M_0(\alpha)\geq M_1(\alpha)\geq M_2(\alpha)\geq\cdots\geq0,
    \label{eq:app_monotonicity}
\end{align}
and, for every $n\geq1$,
\begin{align}
    M_n(\alpha)\leq\prod_{k=1}^{n}\lambda_k^\downarrow(\hat\rho).
    \label{eq:app_spectral_bound}
\end{align}
\end{lemma}

\begin{proof}
Both statements follow from the interlacing of the eigenvalues of a compression with those of the compressed operator.
Fix $\alpha\in\mathbb C$, and let $B_n\coloneqq B_n(\alpha)$ be the matrices of the proof of Theorem~\ref{th:psd}, so that $M_n(\alpha)=\det B_n$.
Since $0\preceq\hat\rho\preceq\hat{\mathds 1}$, the same holds for $\hat\rho_\alpha\coloneqq\hat D^\dagger(\alpha)\hat\rho\hat D(\alpha)$ and hence for each of its compressions $B_n$, whose eigenvalues therefore lie in $[0,1]$.
In particular, $M_n(\alpha)\geq0$ and $M_1(\alpha)\leq1=M_0(\alpha)$.

Since $B_n$ is the leading principal submatrix of $B_{n+1}$, Cauchy's interlacing theorem~\cite[Thm.~4.3.17]{Horn2012-se} gives $\lambda_{k+1}^\downarrow(B_{n+1})\leq\lambda_k^\downarrow(B_n)$ for $1\leq k\leq n$.
Together with $\lambda_1^\downarrow(B_{n+1})\leq1$, this yields
\begin{align}
    M_{n+1}(\alpha)
    &=\lambda_1^\downarrow(B_{n+1})\prod_{k=1}^{n}\lambda_{k+1}^\downarrow(B_{n+1})
    \nonumber
    \\&\leq\prod_{k=1}^{n}\lambda_k^\downarrow(B_n)
    =M_n(\alpha),
\end{align}
which proves Eq.~\eqref{eq:app_monotonicity}.

Since $B_n$ is the compression of the positive compact operator $\hat\rho_\alpha$ to $\operatorname{span}\{\ket0,\ldots,\ket{n-1}\}$, the min--max principle applied to $-\hat\rho_\alpha$~\cite[Sec.~XIII.1]{Reed1978-so} gives $\lambda_k^\downarrow(B_n)\leq\lambda_k^\downarrow(\hat\rho_\alpha)=\lambda_k^\downarrow(\hat\rho)$ for $1\leq k\leq n$, the last equality because $\hat\rho_\alpha$ and $\hat\rho$ are unitarily equivalent.
Multiplying these $n$ inequalities between nonnegative numbers proves Eq.~\eqref{eq:app_spectral_bound}.
\end{proof}

\section{Differential characterization of valid Husimi functions}
\label{app:husimi-admissibility}

In this appendix, we prove Theorem~\ref{th:admissibility}.
Necessity follows directly from the properties of density operators; for sufficiency, we reconstruct a candidate operator from the Bargmann kernel of condition~$\mathrm{(i)}$, and show that conditions~$\mathrm{(ii)}$--$\mathrm{(iv)}$ make it Hermitian, positive semidefinite, and of unit trace.

\begin{proof}[Proof of Theorem~\ref{th:admissibility}]
Suppose first that $Q(\alpha)=\bra{\alpha}\hat\rho\ket{\alpha}$ for a quantum state $\hat\rho$.
Since $\hat\rho$ is bounded, its Fock matrix entries define a jointly entire Bargmann kernel $K$.
The identity $Q(\alpha)=e^{-\alpha\alpha^\ast}K(\alpha^\ast,\alpha)$ then establishes condition~$\mathrm{(i)}$.
Hermiticity of $\hat\rho$ implies that $Q$ is real-valued, giving condition~$\mathrm{(ii)}$.
The coherent-state resolution of the identity and $\operatorname{Tr}\hat\rho=1$ give condition~$\mathrm{(iii)}$.
Finally, the reconstruction formula~\eqref{eq:husimi_reconstruction} identifies $\mathcal M_n[Q]$ with the displaced Fock minors of $\hat\rho$, which are nonnegative because $\hat\rho\succeq0$.
Thus condition~$\mathrm{(iv)}$ also holds.

Conversely, suppose that $Q$ satisfies conditions~$\mathrm{(i)}$--$\mathrm{(iv)}$.
Let $K$ be the jointly entire function of condition~$\mathrm{(i)}$, and define the Fock matrix $(O_{mn})_{m,n\geq0}$ through its Taylor expansion:
\begin{align}
K(z,w)
=
\sum_{m,n\geq0}
\frac{O_{mn}}{\sqrt{m!n!}}\,z^m w^n.
\end{align}
The associated operator $\hat O_Q\coloneqq\sum_{mn}O_{mn}\ket{m}\!\!\bra{n}$ is Bargmann-entire and has Husimi function $Q$.
Condition~$\mathrm{(ii)}$ implies that $K(\alpha^\ast,\alpha)=e^{|\alpha|^2}Q(\alpha)$ is real; comparing Taylor coefficients with those of its complex conjugate gives $O_{mn}=O_{nm}^{\ast}$.
Hence $\hat O_Q$ is Hermitian.

The reconstruction formula identifies $\mathcal M_n[Q]$ with the displaced Fock minors of $\hat O_Q$.
Condition~$\mathrm{(iv)}$ and Theorem~\ref{th:psd} therefore imply $\hat O_Q\succeq0$.

Since $\hat O_Q$ is not yet known to define an operator on $\mathcal H$, we establish the trace identity directly from its Fock matrix.
Since $K$ is jointly entire, its Taylor series converges uniformly on compact subsets of $\mathbb C^2$.
With $\alpha=\sqrt{t}\,e^{i\theta}$, we may therefore integrate term by term in $\theta$ at each fixed $t$, leaving only diagonal terms.
The remaining terms are nonnegative because $O_{nn}\geq0$, so Tonelli's theorem allows their sum to be interchanged with the radial integral.
Condition~$\mathrm{(iii)}$ consequently gives
\begin{align}
1
&=
\int_{\mathbb C}
Q(\alpha)\,\frac{\mathrm d^2\alpha}{\pi}
\nonumber\\
&=
\sum_{n\geq0}
\frac{O_{nn}}{n!}
\int_0^\infty e^{-t}t^n\,\mathrm dt
\nonumber\\
&=
\sum_{n\geq0}O_{nn}.
\end{align}

Since the Fock matrix is positive semidefinite and its diagonal sum is one, it defines a positive trace-class operator on $\mathcal H$ with unit trace.
Thus $\hat O_Q$ is a density operator whose Husimi function is $Q$.

\end{proof}

\section{Differential recurrence for Fock minors}
\label{app:differential_recurrence}

In this appendix, we prove the differential recurrence of Theorem~\ref{th:recurrence}.
The proof relies on a known identity satisfied by bi-directional Wronskians, which we translate into a relation between displaced Fock minors using Lemma~\ref{lem:fock_minors_biwronskians}.
We then show that the recurrence extends to the differential Fock minors of any smooth phase-space function.

\begin{proof}[Proof of Theorem~\ref{th:recurrence}]
Let $K$ be the Bargmann kernel of $\hat O$, and write
$\tau_n\coloneqq\tau_n[K]$.
It was shown in Ref.~\cite{Ma2011-ee}, using the Jacobi identity for determinants, that $\tau_n$ satisfies the bilinear two-dimensional Toda molecule equation
\begin{align}
    \tau_n\partial_z\partial_w\tau_n
    -
    (\partial_z\tau_n)(\partial_w\tau_n)
    =
    \tau_{n+1}\tau_{n-1}.
    \label{eq:toda_biwronskian}
\end{align}

Evaluating along $z=\alpha^*$ and $w=\alpha$, the chain rule gives $\partial_w\tau_n=\partial_\alpha\tau_n$ and $\partial_z\tau_n=\partial_{\alpha^*}\tau_n$.
Hence Eq.~\eqref{eq:toda_biwronskian} becomes
\begin{align}
    \tau_n\partial_\alpha\partial_{\alpha^*}\tau_n
    -
    (\partial_\alpha\tau_n)(\partial_{\alpha^*}\tau_n)
    =
    \tau_{n+1}\tau_{n-1},
    \label{eq:toda_restricted}
\end{align}
where all bi-directional Wronskians are evaluated at $(\alpha^*,\alpha)$.

Lemma~\ref{lem:fock_minors_biwronskians} gives
\begin{align}
    \tau_n[K](\alpha^*,\alpha)
    =
    \left(
        \prod_{k=0}^{n-1}k!
    \right)
    e^{n|\alpha|^2}M_n(\alpha).
\end{align}
Substituting this relation into \eqref{eq:toda_restricted}, using the product rule and $\partial_\alpha\partial_{\alpha^*}|\alpha|^2=1$, we obtain
\begin{align}
    &\left(e^{n|\alpha|^2}M_n\right)
    \partial_\alpha\partial_{\alpha^*}
    \left(e^{n|\alpha|^2}M_n\right)
    \nonumber\\
    &\quad-
    \partial_\alpha\left(e^{n|\alpha|^2}M_n\right)
    \partial_{\alpha^*}\left(e^{n|\alpha|^2}M_n\right)
    \nonumber\\
    &=
    e^{2n|\alpha|^2}
    \Big[
        M_n\partial_\alpha\partial_{\alpha^*}M_n
        -
        (\partial_\alpha M_n)(\partial_{\alpha^*}M_n)
        +
        nM_n^2
    \Big].
    \label{eq:gaussian_bilinear_identity}
\end{align}
Therefore, the left-hand side of
Eq.~\eqref{eq:toda_restricted} is
\begin{align}
    &\left(
        \prod_{k=0}^{n-1}k!
    \right)^2
    e^{2n|\alpha|^2}
    \Big[
        M_n\partial_\alpha\partial_{\alpha^*}M_n
        \nonumber\\
    &\hspace{4em}
        -
        (\partial_\alpha M_n)(\partial_{\alpha^*}M_n)
        +
        nM_n^2
    \Big].
\end{align}

On the other hand, Eq.~\eqref{eq:M_tau_relation} gives
\begin{align}
    \tau_{n+1}
    &=
    n!
    \left(
        \prod_{k=0}^{n-1}k!
    \right)
    e^{(n+1)|\alpha|^2}M_{n+1},
    \\
    \tau_{n-1}
    &=
    \frac{1}{(n-1)!}
    \left(
        \prod_{k=0}^{n-1}k!
    \right)
    e^{(n-1)|\alpha|^2}M_{n-1}.
\end{align}
Consequently,
\begin{align}
    \tau_{n+1}\tau_{n-1}
    =
    n
    \left(
        \prod_{k=0}^{n-1}k!
    \right)^2
    e^{2n|\alpha|^2}
    M_{n+1}M_{n-1}.
\end{align}
Cancelling the common factors in Eq.~\eqref{eq:toda_restricted}
finally yields
\begin{align}
    M_n\partial_\alpha\partial_{\alpha^*}M_n
    \!-\!
    (\partial_\alpha M_n)(\partial_{\alpha^*}M_n)
    \!=\!
    n\left(
        M_{n+1}M_{n-1}\!-\!M_n^2
    \right),
\end{align}
which proves the result.

\end{proof}

Theorem~\ref{th:recurrence} is stated for Fock minors of an operator.
However, the differential Fock minors $\mathcal M_n[Q]$ of Eq.~\eqref{eq:differential_fock_minor} are defined for any smooth function $Q$, whether or not an underlying operator exists.
The following lemma shows that the recurrence still holds in this more general setting, so that the whole hierarchy can be generated from any candidate smooth function $Q$ before even knowing whether it describes a quantum state.

\begin{lemma}[Extension to smooth functions]
\label{lem:smooth_recurrence}
Let $Q$ be a smooth function on $\mathbb C$. For every $n\geq1$, the differential Fock minors of $Q$ satisfy the recurrence:
\begin{align}
    &\mathcal M_n[Q]\,\partial_\alpha\partial_{\alpha^*}\mathcal M_n[Q]
    -\left(\partial_\alpha\mathcal M_n[Q]\right)\left(\partial_{\alpha^*}\mathcal M_n[Q]\right)
    \nonumber\\
    &\qquad=
    n\left(\mathcal M_{n+1}[Q]\,\mathcal M_{n-1}[Q]-\mathcal M_n[Q]^2\right),
\end{align}
at every point of phase space.
\end{lemma}

\begin{proof}
The idea is that the recurrence at a given point only involves finitely many derivatives of $Q$ at that point.
It is therefore enough to check it for a polynomial with the same derivatives, and such a polynomial is the Husimi function of a Bargmann-entire operator.

Fix $n\geq1$ and a point $\alpha_0\in\mathbb C$.
By Eq.~\eqref{eq:differential_fock_minor}, both sides of the recurrence, evaluated at $\alpha_0$, depend on $Q$ only through the derivatives $\partial_{\alpha^*}^{\,p}\partial_\alpha^{\,q}Q(\alpha_0)$ with $0\leq p,q\leq n$.
Let $P(\alpha^*,\alpha)$ be the Taylor polynomial of $Q$ at $\alpha_0$ that has exactly these derivatives, i.e., with
$\partial_{\alpha^*}^{\,p}\partial_\alpha^{\,q}P(\alpha_0)=\partial_{\alpha^*}^{\,p}\partial_\alpha^{\,q}Q(\alpha_0)$ for all $0\leq p,q\leq n$.

The polynomial $P$ is the Husimi function of an operator.
Indeed, since $\bra{\alpha}(\hat a^\dagger)^p\hat a^q\ket{\alpha}=(\alpha^*)^p\alpha^q$, replacing each monomial $(\alpha^*)^p\alpha^q$ of $P$ by the normally ordered product $(\hat a^\dagger)^p\hat a^q$ yields an operator $\hat P$ with $Q_{\hat P}=P$.
Its Bargmann kernel, $e^{zw}P(z,w)$, is entire, so $\hat P$ is Bargmann-entire.
Theorem~\ref{th:recurrence} applies to $\hat P$, and Lemma~\ref{lem:husimi-reconstruction} identifies its Fock minors with $\mathcal M_k[P]$.
The recurrence thus holds for $\mathcal M_k[P]$, and in particular at $\alpha_0$.

Since $P$ and $Q$ have the same derivatives at $\alpha_0$ up to the order involved, the two sides of the recurrence take the same values for $P$ and for $Q$ at that point.
The recurrence therefore holds for $\mathcal M_k[Q]$ at $\alpha_0$.
As $n$ and $\alpha_0$ are arbitrary, this proves the lemma.
\end{proof}

\section{Stellar factorization of the last nonvanishing Fock minor}
\label{app:bargmann_wronskian}

In this appendix, we derive a closed form for the last nonvanishing Fock minor of a finite-rank state, i.e., the minor $M_r$ whose order equals the rank $r$ of the state.
We show that it factorizes into a spectral prefactor, a Gaussian, and the squared modulus of the Wronskian of the stellar functions of the eigenstates, which proves Eq.~\eqref{eq:terminal_minor_stellar_wronskian}.

Let $\hat\rho$ be a rank-$r$ state with spectral decomposition $\hat{\rho}=\sum_{k=0}^{r-1}\lambda_k\ket{\psi_k}\!\bra{\psi_k}$, with $\lambda_k>0$.
We denote by $\psi_k(z)\coloneqq\bra{0}e^{z\hat{a}}\ket{\psi_k}$ the Bargmann representation of $\ket{\psi_k}$, and define
\begin{align}
    \mathcal W_\rho(z)
    \coloneqq&
    \operatorname{Wr}[\psi_0,\ldots,\psi_{r-1}](z)
    \\=&
    \det\!\left[
        \partial_z^i\psi_j(z)
    \right]_{0\leq i,j<r}.
\end{align}

Let $V_{ij}(\alpha)\coloneqq\bra{i}\hat D^\dagger(\alpha)\ket{\psi_j}$ for $0\leq i,j<r$.
Inserting the spectral decomposition of $\hat\rho$, the leading $r\times r$ block of the displaced Fock matrix factorizes as
\begin{align}
    \left[
        \bra{i}
        \hat D^\dagger(\alpha)\hat\rho\hat D(\alpha)
        \ket{k}
    \right]_{0\leq i,k<r}
    =
    V(\alpha)\Lambda V(\alpha)^\dagger,
\end{align}
where $\Lambda=\operatorname{diag}(\lambda_0,\ldots,\lambda_{r-1})$.
Therefore
\begin{align}
    M_{r}(\alpha)
    =
    \left(
        \prod_{j=0}^{r-1}\lambda_j
    \right)
    |\det V(\alpha)|^2 .
    \label{eq:terminal_minor_gram}
\end{align}

It remains to compute $\det V(\alpha)$. 
Using the normally ordered form of the displacement operator, one obtains
\begin{align}
    \bra{0}e^{z\hat a}
    \hat D^\dagger(\alpha)
    \ket{\psi_j}
    =
    e^{-|\alpha|^2/2}
    e^{-\alpha z}
    \psi_j(z+\alpha^*).
\end{align}
Since
\begin{align}
    V_{ij}(\alpha)
    =
    \frac{1}{\sqrt{i!}}
    \left.
    \partial_z^i
    \left(
        e^{-|\alpha|^2/2}
        e^{-\alpha z}
        \psi_j(z+\alpha^*)
    \right)
    \right|_{z=0},
\end{align}
we find
\begin{align}
    \det V(\alpha)
    =
    \frac{e^{-r|\alpha|^2/2}}
    {\sqrt{\prod_{k=0}^{r-1}k!}}\,
    \mathcal W_\rho(\alpha^*).
    \label{eq:detV_wronskian}
\end{align}
Here we used the multiplicative Wronskian identity~\cite[Eq.~(14)]{Braunstein1998}, $\operatorname{Wr}[g\phi_0,\ldots,g\phi_{r-1}]=g^r\operatorname{Wr}[\phi_0,\ldots,\phi_{r-1}]$, followed by evaluation at $z=0$.

Combining Eqs.~\eqref{eq:terminal_minor_gram} and~\eqref{eq:detV_wronskian} gives
\begin{align}
    M_{r}(\alpha)
    =
    \left(
        \prod_{k=0}^{r-1}
        \frac{\lambda_k}{k!}
    \right)
    e^{-r|\alpha|^2}
    \left|
        \mathcal W_\rho(\alpha^*)
    \right|^2,
\end{align}
which is Eq.~\eqref{eq:terminal_minor_stellar_wronskian}.

Under a change of orthonormal basis of $\operatorname{supp}(\hat\rho)$, represented by a unitary matrix
$U$, the Wronskian matrix is simply multiplied by $U$.
Consequently, $\mathcal W_\rho$ is multiplied by $\det U$, which is a constant phase that drops out upon taking the modulus squared.
Thus, apart from the factor $\prod_j\lambda_j$, the Fock minor associated with the rank of $\hat{\rho}$ is solely determined by its support. 
\end{document}